\documentclass{IEEEtran4PSCC}

\IEEEoverridecommandlockouts 

\usepackage{graphicx}      

\usepackage{commath}
\usepackage{amsfonts}
\usepackage{mathtools}
\usepackage{steinmetz}		
\usepackage{supertabular}	
\usepackage{tabularx}		
\usepackage{algorithm}		
\usepackage{algorithmic}		
\usepackage{amsmath,amssymb}
\usepackage{amsthm}
\newtheorem{rem}{Remark}
\newtheorem{defi}{Definition}
\newtheorem{prop}{Proposition}
\newtheorem{thm}{Theorem}
\theoremstyle{proof}

\usepackage{subfig}
\usepackage{enumitem} 		
\usepackage{soul}

\usepackage{cite}

\ifCLASSINFOpdf
\else
\fi
\usepackage{array}
\begin{document}
%
\title{A Two-Step Distribution System State Estimator with Grid Constraints and Mixed Measurements}

 \author{\IEEEauthorblockN{Miguel Picallo\IEEEauthorrefmark{1},
Adolfo Anta\IEEEauthorrefmark{2},
Bart de Schutter\IEEEauthorrefmark{1} and
Ara Panosyan\IEEEauthorrefmark{3}}
 \IEEEauthorblockA{\IEEEauthorrefmark{1}Delft Center for Systems and Control, Delft University of Technology \\
Delft, The Netherlands \\
m.picallocruz,b.deschutter@tudelft.nl}
\IEEEauthorblockA{\IEEEauthorrefmark{2}Austrian Institute of Technology \\
Vienna, Austria}
\IEEEauthorblockA{\IEEEauthorrefmark{3}General Electric Global Research \\
Munich, Germany}

\thanks{This project has received funding from the European Union's Horizon 2020 research and innovation programme under the Marie Skł{\l}odowska-Curie grant agreement No 675318 (INCITE).}}


%


\maketitle

\begin{abstract}
In this paper we consider the problem of State Estimation (SE) in large-scale, 3-phase coupled, unbalanced distribution systems. More specifically, we address the problem of including mixed real-time measurements, synchronized and unsynchronized, from phasor measurement units and smart meters, into existing SE solutions. We propose a computationally efficient two-step method to update a prior solution using the measurements, while taking into account physical constraints caused by zero-injection buses. We test the method on a benchmark test feeder and show the effectiveness of the approach.
\end{abstract}



%
\IEEEpeerreviewmaketitle

\section{Introduction}
The operation of a power network requires accurate monitoring of its state: bus voltages, line currents, consumption, and generation. This is specially relevant in transmission networks where volatile and distributed generation and consumption cause bidirectional power flow and voltage drops. Normally, State Estimation (SE) consists in estimating the bus voltage phasors using the power flow equations derived from the structure of the network, represented by the admittance matrix. SE is typically performed by taking several measurements and solving a weighted least-squares problem using an iterative approach like Newton-Raphson \cite{abur2004power, huang2012state, monticelli2000electric, hayes2014state, holten1988comparison}. 

Typically distribution networks used to have a simple radial structure with a single source bus, the point of common coupling (PCC), connected to the grid and injecting power, and therefore SE was not as necessary as in transmission networks. This is changing given the increasing number of sources of distributed generation like PV-panels, electrical vehicles, batteries, etc. injecting power into the network and causing bidirectional power flow. Consequently, SE becomes necessary in distribution networks. One of the major difficulties for its implementation lies in the structural difference between transmission and distribution networks: in the latter, there are coupled phases, unbalanced loads, and lower $X/R$ ratios. Therefore, the power flow equations need to be solved for the 3 phases simultaneously, and approximations like the fast decoupled power flow \cite{stott1974fast,garcia1979fast} cannot be used. Another limitation is the lack of sufficient real-time measurement units in distribution networks: although Phasor Measurement Units (PMUs) and smart meters are being deployed, their high cost \cite{madani2011pmu} prevents installing the required number of sensors to make the system fully observable.

Some methods have been proposed to perform SE by exploiting the network structure and the nature of the measurements, such as branch-current-based SE \cite{baran1994state, kelley1995branch, lin2001highly}. Other methods use the few measurements available to improve the SE accuracy of previous methods based on load assumptions only \cite{schenato2014bayesian, zhou2006alternative,hu2011belief}. Yet, most of this work is limited to transmission networks, small-scale distribution networks, or single-phase networks. Some recent work extends it to more general networks, but assumes availability of measurements at every bus \cite{bolognani2014state}, which at the moment is not the case for most distribution networks. Additionally, when traditional methods \cite{abur2004power, holten1988comparison} deal with network constraints, such as zero-injection buses, become computationally expensive as the size of the network increases. This may be a limitation in large distribution systems when fast SE results may be required every few seconds. In the current paper, we will address theses issues.

In particular, our main contributions can be summarized as follows: First, we include network equality constraints more efficiently by a dimension reduction similar to \cite{guo2013efficient}. Then, we extend some recently proposed methods that split the problem \cite{schenato2014bayesian}, to large-scale, 3-phase coupled, unbalanced and constrained distribution systems, with synchronized and unsynchronized real-time measurements affected by both magnitude and angle noises. 

The rest of the paper is structured as follows: Section \ref{sec:grid} presents some background about power networks. Section \ref{sec:info} discusses the different types of measurements available for SE. Section \ref{sec:SEbasic} presents the standard methodology for SE. Section \ref{sec:methods} presents our contributions: it proposes the linear dimension reduction method and extends recent SE methodologies to split the problem. Section \ref{sec:sim} defines the test case simulation and shows its results. Finally, Section \ref{sec:conc} presents the conclusions and proposes topics for future work.


\section{Distribution System Model}\label{sec:grid}
A distribution system consists of buses, where power is injected or consumed, and branches, each connecting two buses. This system can be modeled as a graph $\mathcal{G}=(\mathcal{V},\mathcal{E},\mathcal{W})$ with nodes $\mathcal{V}=\{1,...,N_\text{bus}\}$ representing the buses, edges $\mathcal{E}=\{(v_i,v_j)\mid v_i,v_j \in \mathcal{V}\}$ representing the branches, and edge weights $\mathcal{W}=\{w_{i,j}\mid (v_i,v_j) \in \mathcal{E}, w_{i,j}\in \mathbb{C}\}$ representing the admittance of a branch, which is determined by the length and type of the line cables.

In 3-phase networks buses may have up to 3 phases, so that the voltage $V_i$ at bus $i$ lives in $\mathbb{C}^{\leq 3}$ (and the edge weights $w_{i,j}\in \mathbb{C}^{\leq 3 \times \leq 3}$). The state of the network is then typically represented by the vector bus voltages $V_\text{bus}=[V_\text{src}^T \; V^T]^T \in \mathbb{C}^{N+3}$, where $V_{\text{src}} \in \mathbb{C}^3$ denotes the known voltage at the source bus, and $V \in \mathbb{C}^N$ the voltages in the non-source buses, where $N$ depends on the number of buses and phases per bus.


Using the Laplacian matrix $Y \in \mathbb{C}^{(N+3) \times (N+3)}$ of the weighted graph $\mathcal{G}$, called admittance matrix \cite{abur2004power}, the power flow equations to compute the currents $I$ and the power loads $S$ are:
\begin{equation}\label{eq:PFeq}\arraycolsep=1pt
\begin{array}{c}
\left[\begin{array}{c} I_{\text{src}} \\ I \end{array}\right] = 
Y\left[\begin{array}{c} V_{\text{src}} \\ V \end{array}\right] =
\left[\begin{array}{cc} Y_\text{a} & Y_\text{b} \\ Y_\text{c} & Y_\text{d} \end{array}\right]
\left[\begin{array}{c} V_{\text{src}} \\ V \end{array}\right],  \; S = \text{diag}(\bar{I})V
\end{array}
\end{equation}
where $\bar{(\cdot)}$ denotes the complex conjugate, $\text{diag}(\cdot)$ represents the diagonal operator, converting a vector into a diagonal matrix. Separating $Y$ in blocks according to the indices of the source bus $V_{\text{src}}$, see \eqref{eq:PFeq}, the voltage $V$ for the non-source buses can be rewritten as:
\begin{equation}\label{eq:PFit}
\begin{array}{c}
V =  Y_\text{d}^{-1}I-Y_\text{d}^{-1}Y_\text{c}V_{\text{src}} \\[0.1cm]
V_0=V\mid_{I=0}=-Y_\text{d}^{-1}Y_\text{c}V_{\text{src}}
\end{array}
\end{equation}
with $V_0$ denoting the voltage without loads

\section{Measurements} \label{sec:info}
Several sources of information can be available to solve the SE problem:
\begin{enumerate}[leftmargin=*]
\item \textit{Pseudo-measurements}: They can be load estimations for every hour, based on predictions and known installed load capacity at every bus, and represented by $S_\text{psd}$. Since they are estimations rather than actual measurements, their noise is modeled with a relative large standard deviation (a typical value can be $\sigma_\text{psd} \approx 50\%$ \cite{schenato2014bayesian}):
\begin{equation*}
S = S_\text{psd} +S_\text{psd} \omega_\text{psd} = (P_\text{psd}+jQ_\text{psd})(1+\omega_\text{psd})
\end{equation*}
where $j$ is the imaginary unit, $\omega_\text{psd}\sim \mathcal{N}(0,\sigma_\text{psd}I_{\text{d},N})$, $I_{\text{d},N}$ is the identity matrix with size $N\times N$, and $P=\Re\{S\},Q=\Im\{S\}$ denote the real and imaginary parts of the power. The noise covariance of these measurements is then:
\begin{equation}\arraycolsep=1pt
\Sigma_{\text{psd}}=\mathbb{E}[\;\abs{S-S_\text{psd}}^2]=\sigma_\text{psd}^2 \text{diag}(\;\abs{S_\text{psd}}^2)
\end{equation}
where $\mathbb{E}[\cdot]$ denotes expectation, $\abs{\cdot}$ element-wise magnitudes of a complex vector, and $(\cdot)^2$ element-wise square.

\item \textit{Virtual measurements}: They are buses with no loads connected, zero injection buses. They can be modeled as physical constraints for the voltage states by defining the set of indices of zero-injection buses $\varepsilon=\{i,\cdots,j\}$:
\begin{equation}\label{eq:Scons}
(S)_{\varepsilon}=0, \; (I)_{\varepsilon}=0
\end{equation} 
where $(\cdot)_{\varepsilon}$ denotes the elements at indices in ${\varepsilon}$.

\item \textit{Real-time measurements}: They can be a range of voltages and currents measurements from PMUs, smart meters, or conventional remote terminal units. There are two possible kinds: GPS-synchronized sensors measuring magnitude and phase angle, and unsynchronized sensors providing only magnitude. We model the noises with a low standard deviation for magnitude and angle, $\sigma_\text{mag} \approx 1\%$ and $\sigma_\text{ang} \approx 0.01\text{ rad}$ respectively, according to the IEEE standard for PMUs \cite{martin2008exploring}.

Synchronized measurements can be expressed using a linear approximation (see appendix) with magnitude and angle noise caused by the measurements and by an imperfect synchronization. For a number $N_\text{measL}$ of measurements $z_\text{measL} \in \mathbb{C}^{N_\text{measL}}$ we have:
\begin{equation}\label{eq:Lmeas}
\begin{array}{c}
z_\text{measL} \approx  C_\text{measL}V + \text{diag}(C_\text{measL}V)(\omega_{\text{mag}} + j\omega_{\text{ang}})
\end{array}
\end{equation}
where $C_\text{measL}$ is the matrix mapping state voltages to measurements, $\omega_\text{mag} \sim \mathcal{N}(0,\sigma_\text{mag}I_{\text{d},N_\text{measL}})$, $\omega_\text{ang} \sim \mathcal{N}(0,\sigma_\text{ang}I_{\text{d},N_\text{measL}})$. Since $\sigma_\text{mag}$ and $\sigma_\text{ang}$ are similar \cite{martin2008exploring}, for simplicity we will use the same value for both: $\sigma_\text{meas}=\sigma_\text{ang}=\sigma_\text{mag}$. {\color{black} Then, for measurement $i$ at phase $l$ of bus $n$ we have:
\begin{equation}\label{eq:LmeasMap}\arraycolsep=1pt\begin{array}{l}
(C_\text{measL}V)_i = (C_\text{measL})_{i,\bullet}V= \\[0.1cm]
\left\lbrace \begin{array}{ll}
V_{n_l} & \mbox{for a voltage measurement} \\[0.0cm]
(Y)_{n_l,\bullet}V & \mbox{for a current measurement} \\[0.0cm]
(Y)_{n_l,m_l}(V_{n_l}-V_{m_l}) & \mbox{for a branch-current $i \to m$} \\[-0.0cm]
&  \mbox{measurement}
\end{array} \right.
\end{array}
\end{equation}
where $(\cdot)_{i,\bullet}$ denotes row $i$. 

Unsynchronized measurements have a nonlinear relation $C_\text{measNL}(V)$ with the voltages $V$. For $N_\text{measNL}$ measurements $z_\text{measNL} \in \mathbb{R}^{N_\text{measNL}}$ we have:
\begin{equation}\label{eq:NLmeas}
z_{\text{measNL}} =C_\text{measNL}(V)+\text{diag}(C_\text{measNL}(V))\omega_{\text{measNL}} 
\end{equation}
with $\omega_{\text{measNL}}\sim \mathcal{N}(0,\sigma_\text{meas}I_{\text{d},N_\text{measNL}})$. For measurement $i$ at phase $l$ of bus $i$ we have:
\begin{equation}\label{eq:NLmeasMap}\arraycolsep=1pt\begin{array}{l}
(C_\text{measNL}(V))_i =\\[0.1cm] 
\left\lbrace \begin{array}{ll}
\abs{V_{n_l}} & \mbox{for a voltage measurement} \\[0.0cm]
\abs{(Y)_{n_l,\bullet}V} & \mbox{for a current measurement} \\[0.0cm]
\abs{(Y)_{n_l,m_l}(V_{n_l}-V_{m_l})} & \mbox{for a branch-current $i\to m$} \\[-0.0cm]
&  \mbox{measurement}
\end{array} \right.
\end{array}
\end{equation}}
Since the measurement noises in \eqref{eq:Lmeas} and \eqref{eq:NLmeas} are small according to the PMUs standard \cite{martin2008exploring}, their covariance matrices can be approximated using the measurements instead of the actual value:
\begin{equation*}\arraycolsep=1pt
\begin{array}{rl}
\Sigma_\text{measL} & =  \mathbb{E}[\;\abs{z_{\text{measL}}-C_\text{measL}V}^2] 
\\[0.1cm]  
& = 2\sigma_\text{meas}^2 \text{diag}(\:\abs{C_\text{measL}V}^2)
\approx 2\sigma_\text{meas}^2 \text{diag}(\:\abs{z_\text{measL}}^2) \\[0.1cm]
\Sigma_\text{measNL} & = \mathbb{E}[\;\abs{z_{\text{measNL}}-C_\text{measNL}(V)}^2] 
\\[0.1cm] & = \sigma_\text{meas}^2 \text{diag}(C_\text{measNL}(V)^2) 
\approx \sigma_\text{meas}^2 \text{diag}(z_\text{measNL}^2)
\end{array}
\end{equation*}
\end{enumerate}

\section{Standard Methodology for State Estimation}\label{sec:SEbasic}
The standard methodology for SE computes the maximum likelihood estimation by solving a constrained nonlinear weighted least-squares problem with all measurements and estimations: $(S_\text{psd})_{\varepsilon^\text{c}},z_\text{measL},z_\text{measNL}$, their noise covariance, and the system constraints $(S_\text{psd})_{\varepsilon}=0$, where $\varepsilon^\text{c}$ denotes the complement of $\varepsilon$ within the set $\{1,\dots,N\}$. Typically, this problem is solved using the Newthon-Raphson method \cite{abur2004power}. However, since this requires the use of gradients and the power flow formulas in \eqref{eq:PFeq} are not holomorphic, i.e. complex differentiable, the problem is typically solved in real variables using a polar representation $V_\text{polar} = [\: \abs{V}^T \; \angle{V}^T]^T$:
\begin{equation}\label{eq:WLSpolar}
\min_{V_\text{polar}}\norm{z-h(V_\text{polar})}_{W^{-1}} \mbox{ s.t. } g(V_\text{polar}) = 0 
\end{equation} 
where $\norm{x}_{A}^2=x^TAx$ for some $A\succ 0$, $z$ is the vector of measurements in rectangular coordinates: 
\begin{equation*}
z=[
(P_{\text{psd}})_{\varepsilon^c}^T, 
(Q_{\text{psd}})_{\varepsilon^c}^T, 
\Re\{z_\text{measL}\}^T, 
\Im\{z_\text{measL}\}^T, 
z_\text{measNL}^T]^T
\end{equation*}
and $h(\cdot)$ is a nonlinear function mapping $V_\text{polar}$ to measurements. The function $g(\cdot)$ indicates the zero loads:
$g(V_\text{polar})=[(P_\text{psd})_{\varepsilon}^T$, $(Q_\text{psd})_{\varepsilon}^T ]^T$,
and the weight matrix $W^{-1}$ is the inverse of the measurement noises covariances in rectangular variables:
\begin{equation*}\arraycolsep=1pt
W=\text{diag}(\Sigma_\text{rect,psd},\Sigma_\text{rect,measL} , \Sigma_\text{measNL})
\end{equation*}
where here we use the function $\text{diag}(\cdot)$ to denote the operator converting several matrices into a single block-diagonal matrix, and the suffix $(\cdot)_\text{rect}$ stands for rectangular coordinates:
\begin{equation*}\arraycolsep=1pt \begin{array}{c}
\Sigma_\text{rect,measL} \approx \sigma_\text{meas}^2 \text{diag}(\abs{z_\text{measL}}^2,\abs{z_\text{measL}}^2))\\[0.1cm]
\Sigma_\text{rect,psd} = \sigma_\text{psd}^2 \text{diag}\left( \left[\begin{array}{cc}
(P_\text{psd})_{\varepsilon^\text{c}}^2 & (P_\text{psd}Q_\text{psd})_{\varepsilon^\text{c}} \\
(P_\text{psd}Q_\text{psd})_{\varepsilon^\text{c}} & (Q_\text{psd})_{\varepsilon^\text{c}}^2
\end{array}\right] \right)
\end{array}
\end{equation*}
\\
In the literature \cite{abur2004power, monticelli2000electric}, the constraints in \eqref{eq:WLSpolar} are typically included in two ways: 
\begin{enumerate}[leftmargin=*]
\item either using Lagrangian multipliers $\lambda$, which increases the size of the optimization variables to $[V_\text{polar}^T \; \lambda^T]^T \in \mathbb{R}^{2N+2\abs{\varepsilon}}$;
\item or including the equality constraints as measurements with a noise with a very small standard deviation, so that their corresponding weights become very large with respect to other measurements. However, this produces an ill-conditioned weights matrix, which may cause numerical problems. 
\end{enumerate}
In order to efficiently solve either 1) or 2), different methods like the orthogonal, hybrid, or Hachtel methods are typically used \cite{abur2004power, monticelli2000electric, holten1988comparison, korres2010robust, simoes1981robust}. However, these iterative approaches increase the computation time, which may compromise the use of these SE techniques for real-time operation of large distribution networks. Therefore, we propose a more computationally efficient method appropriate to the problem.

\section{Proposed Methods for State Estimation}\label{sec:methods}
In this section we propose new SE methods that significantly improve the standard methodology for \eqref{eq:WLSpolar}. First, we present a method that simplifies the problem and reduces its complexity, while obtaining the same solution; then, we propose another method that significantly reduces the computation cost while achieving approximately the same accuracy.

\subsection{Weighted least-squares in subspace}\label{subsec:SEWLS}
Here we propose an alternative to Lagrangian multipliers in order to include the constraints using a dimension reduction similar to \cite{guo2013efficient}. Therefore, we use a linear transformation that produces feasible solutions. {\color{black} For convenience,} we first define the space of feasible solutions of \eqref{eq:WLSpolar} as:
\begin{equation}\label{eq:Icons}
\begin{array}{l}
\{V\mid I_{\varepsilon}=(Y_\text{d})_{\varepsilon}V + (Y_\text{c})_{\varepsilon,\bullet}V_{\text{src}}=0\} \\[0.1cm]
= \{V\mid \exists x \in \mathbb{C}^{N-\mid \varepsilon \mid} \mbox{ s.t. } V=Fx+V_\text{p}\} \\[0.1cm]
\end{array}
\end{equation}
where $(Y_\text{d})_{\varepsilon,\bullet} \in \mathbb{C}^{\abs{\varepsilon}\times N}$ denotes the rows of $Y_\text{d}$ corresponding to indeces in $\varepsilon$; $V_\text{p}$ is a solution satisfying $I_{\varepsilon}=0$, and $F$ is a basis of the nullspace of $(Y_\text{d})_{\varepsilon,\bullet}$:
\begin{equation}\label{eq:F}
(Y_\text{d})_{\varepsilon,\bullet}F=0
\end{equation} 
Therefore, $V_0$ is chosen for $V_\text{p}$, and $F$ is built using an orthonormal basis of the kernel of $(Y_\text{d})_{\varepsilon,\bullet}$ e.g. by computing the singular value decomposition of $(Y_\text{d})_{\varepsilon,\bullet}=\mathcal{U}S\mathcal{V}^*$, where $(\cdot)^*$ denotes complex conjugate transpose, and taking the columns of $\mathcal{V}$ with zero singular value. Then, $F=\text{ker}((Y_\text{d})_{\varepsilon,\bullet})$ and $F^*F=I_\text{d}$, where $\text{ker}(\cdot)$ denotes the kernel subspace. If we consider the rectangular representation of voltages $V_\text{rect}=[\Re\{V\}^T \; \Im\{V\}^T]^T$ instead of the polar representation $V_\text{polar}$, the subspace equations can still be represented linearly in real variables: 
\begin{equation}\label{eq:Iconsrect}\arraycolsep=1pt
\begin{array}{rl}
\{V_\text{rect} \mid & \Re\{I_{\varepsilon}\}=0, \Im\{I_{\varepsilon}\}=0 \} \\[0.1cm] 
= \Bigg\{\hspace{-0.05cm}V_\text{rect} \Bigg| &
\left[\begin{array}{cr}
\Re\{(Y_\text{d})_{\varepsilon,\bullet}\} & -\Im\{(Y_\text{d})_{\varepsilon,\bullet}\} \\ 
\Im\{(Y_\text{d})_{\varepsilon,\bullet}\} & \Re\{(Y_\text{d})_{\varepsilon,\bullet}\}
\end{array}\right]
V_\text{rect} \\
& +\left[\begin{array}{c}\Re\{(Y_\text{c})_{\varepsilon,\bullet}V_{\text{src}}\} \\ \Im\{(Y_\text{c})_{\varepsilon,\bullet}V_{\text{src}}\} \end{array}\right]=0\Bigg\}
\\[0.4cm]
= \{V_\text{rect} \mid &  \exists \tilde{x} \in \mathbb{C}^{2N-2\mid \varepsilon \mid} \mbox{ s.t. } V_\text{rect}=\tilde{F}\tilde{x}+\tilde{V}_\text{p}\}
\end{array}
\end{equation}
where now
\begin{equation*}\arraycolsep=1pt
\begin{array}{l}
\tilde{V}_\text{p}= V_{\text{rect},0}= \left[\begin{array}{c}
\Re\{V_0\} \\ \Im\{V_0\}
\end{array}\right]\hspace{-0.1cm} \\[0.1cm]
\tilde{F} = \text{ker}\left( \left[\begin{array}{cr}
\Re\{(Y_\text{d})_{\varepsilon,\bullet}\} & -\Im\{(Y_\text{d})_{\varepsilon,\bullet}\} \\ \Im\{(Y_\text{d})_{\varepsilon,\bullet}\} & \Re\{(Y_\text{d})_{\varepsilon,\bullet}\}
\end{array}\right] \right)
\end{array}
\end{equation*}
Then problem \eqref{eq:WLSpolar} becomes:
\begin{equation}\label{eq:WLSrect}
\tilde{F}\left(\arg\min_{\tilde{x}} \norm{z-\tilde{h}(\tilde{F}\tilde{x}+\tilde{V}_\text{p})}_{W^{-1}}\right)+V_{\text{rect},0}
\end{equation} 
where $\tilde{h}(\cdot)$ is like $h(\cdot)$ but in rectangular coordinates, so that $\tilde{h}(V_\text{rect})=h(V_\text{polar})$. This again can be solved using the Newton-Raphson method with iterations:
\begin{equation}
\tilde{x}_{k+1} = \tilde{x}_k + \Delta \tilde{x}_k \\[0.1cm]
\end{equation} 
where
\begin{equation}\label{eq:WLSLinSys}
\begin{array}{c}
(H_k^T W^{-1}H_k)\Delta \tilde{x}_k = H_k^TW^{-1}(z-\tilde{h}(\tilde{F}\tilde{x}_k+\tilde{V}_\text{p})) \\[0.1cm]
H_k = \nabla_{\tilde{x}} \tilde{h}(\tilde{F}\tilde{x}+\tilde{V}_\text{p})\mid_{\tilde{x}_{k}} = \nabla_{V_\text{rect}}\tilde{h}(V_\text{rect}) \mid_{V_{\text{rect},k}}\tilde{F}
\end{array}
\end{equation} 
With the method in \eqref{eq:WLSrect}, we have eliminated the equality constraints by embedding the solution into a smaller subspace of feasible solutions. This simplifies the method and reduces its computational cost with respect to using Lagrangian multipliers or large weights \cite{abur2004power,holten1988comparison}, since we reduce the size of the optimization variables from $[V_\text{polar}^T \; \lambda^T]^T \in \mathbb{R}^{2N+2\abs{\varepsilon}}$ or $V_\text{polar} \in \mathbb{R}^{2N}$ respectively, to $\tilde{x} \in \mathbb{R}^{2N-2\abs{\varepsilon}}$. Additionally, the gain matrix $(H_k^T W^{-1}H_k)$ will not be ill-conditioned, since all weights have comparable magnitudes; and $F$,$\tilde{F}$ can be computed offline. Moreover, with method \eqref{eq:WLSrect} the constraints will be satisfied by construction, which in not guaranteed if using large weights. Oppositely to \cite{guo2013efficient}, we keep the rectangular representation, which keeps the transformation linear, and choose a different representation of the subspace using the orthonormal bases $F$, $\tilde{F}$, in order to compare the method with the later ones in \eqref{eq:WLSupdate}, \eqref{eq:solSEPFUp} and \eqref{eq:solSEPFUpNL}, when transforming the covariance matrices $\Sigma_\text{prior}$ and $\Sigma_\text{post}$.

\subsection{Power flow solution plus optimal linear update}\label{subsec:SEPFUp}
Using Newthon-Raphson for \eqref{eq:WLSpolar} and \eqref{eq:WLSrect} requires to solve in each iteration the linear system \eqref{eq:WLSLinSys} with approximately same size $N$ as the network. Consequently, this can become computationally expensive when applied to large distribution systems. Therefore, we propose an alternative method that splits the problem into two parts as in \cite{schenato2014bayesian}: first, we use only the constraints and the load estimates $S_\text{psd}$ available beforehand to solve the power flow problem offline and to obtain a prior solution defined as $V_ \text{prior}$; then, we update in real-time the solution using the measurements $z_\text{measL},z_\text{measNL}$ to compute a better estimate, represented by $V_\text{post}$. We also prove that two alternatives for the update, the minimum-variance and the maximum-likelihood estimator, are equal under certain conditions.

\subsubsection{Obtaining a prior solution}\label{subsec:PF}
\begin{enumerate}[leftmargin=*]
\item[a)] The prior estimate can be computed using the previous method in \eqref{eq:WLSrect} without real-time measurements: $z=[(P_{\text{psd}})_{\varepsilon^c}^T \: (Q_{\text{psd}})_{\varepsilon^c}^T]^T$ and $\tilde{h}(V_\text{rect})=[ P_{\varepsilon^\text{c}}(V_\text{rect})^T \: Q_{\varepsilon^\text{c}}(V_\text{rect})^T]^T$, where $P(V_\text{rect})$ denotes $P$ as a function of $V_\text{rect}$. This method will estimate real and imaginary parts $V_\text{rect,prior}$, with estimation error covariance:
\begin{equation}\label{eq:sigmaWLS}\arraycolsep=1pt 
\Sigma_{\text{rect,prior}}
\approx \tilde{F}(H_k^T W^{-1}H_k)^{-1}\tilde{F}^T
\end{equation}
\item[b)] An alternative approach using complex variables would be the method proposed in \cite{schenato2014bayesian}, after extending it to a 3-phase coupled, unbalanced system. Using the voltage under no loads $V_0$ as initial value in an iterative method inspired by \eqref{eq:PFit}, for the $k$-th iteration we have:
\begin{equation}\label{eq:linapPF}
V_{k+1} = Y_\text{d}^{-1} (\text{diag}(\bar{V}_k))^{-1} \bar{S} +V_0
\end{equation} 
The estimation error covariance can be approximated by the covariance after the first iteration, which corresponds to the linear approximation in \cite{bolognani2016existence}: 
\begin{equation}\label{eq:sigmaapprox}
\Sigma_\text{prior} \approx Y_\text{d}^{-1} (\text{diag}(\bar{V}_0))^{-1} \Sigma_{\text{psd}} (\text{diag}(V_0))^{-1} (Y_\text{d}^{-1})^*
\end{equation}
\begin{rem}
The estimated current at any iteration $k$ is 
\begin{equation}
I_{k}= (\text{diag}(\bar{V}_{k-1}))^{-1} \bar{S}
\end{equation}
If the loads satisfy $(S)_ \varepsilon=0$, then $(I_{k})_\varepsilon =0$. This means that there exists a vector $x\in \mathbb{C^{N-\abs{\varepsilon}}}$ s.t. $V_k=Fx+V_0$ and thus $V_k$ is feasible for all $k$ with respect to \eqref{eq:Icons}. 
\end{rem}
\begin{rem}
Method b) can be extended to real variables to get $\Sigma_\mathrm{prior}$ in rectangular coordinates:
\begin{equation}\arraycolsep=1pt
V_{\mathrm{rect,}k+1} =
B_k \left[\begin{array}{c} \Re\{S_{\mathrm{psd}}\} \\ \Im\{S_{\mathrm{psd}}\}\end{array}\right]
+V_{\mathrm{rect,}0}
\end{equation}
with 
\begin{equation*}\arraycolsep=1pt
\begin{array}{c}
B_k = \left[\begin{array}{cr} \Re\{Y_\mathrm{d}^{-1}\} & -\Im\{Y_\mathrm{d}^{-1}\} \\ \Im\{Y_\mathrm{d}^{-1}\} & \Re\{Y_\mathrm{d}^{-1}\}\end{array}\right]
\Bigg(\mathrm{diag}\Bigg(\left[\begin{array}{c}\abs{V_k}^2 \\ \abs{V_k}^2\end{array}\right]\Bigg)\Bigg)^{-1} \\[0.4cm]
\cdot 
\left[\begin{array}{cr} \mathrm{diag}(\Re\{V_{k}\}) & \mathrm{diag}(\Im\{V_{k}\}) \\ \mathrm{diag}(\Im\{V_{k}\}) & -\mathrm{diag}(\Re\{V_{k}\}) \end{array}\right]
\end{array}
\end{equation*}
and estimation error covariance:
\begin{equation}\arraycolsep=1pt
\Sigma_{\mathrm{rect,prior}} = 
B_0\sigma_\mathrm{psd}^2
\mathrm{diag}\left(\left[
\begin{array}{c}\Re\{S_{\mathrm{psd}}\}^2 \\ \Im\{S_{\mathrm{psd}}\}^2 \end{array}
\right]\right)
B_0^T 
\end{equation}
\end{rem}
\end{enumerate}
The advantage of approach b) with respect to a) is that b) can be solved in complex variables. In both cases, since the load estimates $S_\text{psd}$ are known beforehand, this part of the problem can be computed offline, taking this iterative computational cost outside the real-time operation of the system. 

\subsubsection{Updating with only synchronized measurements}\label{subsec:updatePF}~\\
If there are not any unsynchronized measurements, the mapping from the state variables to the real-time measurements \eqref{eq:LmeasMap} is linear and thus holomorphic in $x\in\mathbb{C}^{N-\abs{\varepsilon}}$. Consequently, the update can be solved using complex variables to obtain a linear expression. Different alternative methods are possible; here we present two of them and later show their equivalence:
\begin{enumerate}[leftmargin=*]
\item[a)] If the prior solution $V_\text{prior}$ is feasible, then there exists a vector $x_\text{prior}$ s.t. $V_\text{prior}=Fx_\text{prior}+V_0$. It can be computed as $x_\text{prior}=F^*(V_\text{prior}-V_0)$, with estimation error covariance $\Sigma_{x_\text{prior}}=F^*\Sigma_{\text{prior}}F$, or $x_\text{rect,prior}=\tilde{F}^T(V_\text{rect,prior}-V_{\text{rect,}0})$ for real variables. Then the maximum likelihood update can be computed using weighted least-squares:
\begin{equation}\label{eq:WLSupdate}
\arraycolsep=1pt
V_\text{post}=V_0+F\arg\min_x\norm{\left[\begin{array}{c}
x - x_\text{prior} \\ C_\text{measL}(Fx+V_0)-z_\text{measL}
\end{array}\right]}_{W_\text{post}^{-1}}^2
\end{equation}
where $W_\text{post}^{-1}=(\text{diag}(F^*\Sigma_{\text{prior}}F , \Sigma_\text{measL}))^{-1}$ is the weight matrix. Taking derivatives at \eqref{eq:WLSupdate} we obtain the solution:
\begin{equation}\label{eq:WLSupdateSol}
\arraycolsep=1pt
\begin{array}{rl}
V_\text{post}=&  F((F^*\Sigma_{\text{prior}}F)^{-1}+F^*C_\text{measL}^*\Sigma_\text{measL}^{-1}C_\text{measL}F)^{-1} \\[0.1cm]
& \cdot((F^*\Sigma_{\text{prior}}F)^{-1}x_\text{prior}+F^*C_\text{measL}^*\Sigma_\text{measL}^{-1} \\[0.1cm]
& \cdot(z_\text{measL}-C_\text{measL}V_0)) +V_0
\end{array}
\end{equation}

\begin{rem}
Considering the first-order (and thus linear) approximation of the measurement function $\tilde{h}(\cdot)$, in \cite{zhou2006alternative} it is proven that solving the weighted least-squares problem in one step \eqref{eq:WLSrect} and in two steps (i.e. weighted least-squares for power flow plus update \eqref{eq:WLSupdate}), produce exactly the same solution. Consequently, we can expect a similar accuracy between the one-step and the two-steps methods when applied to the original nonlinear problem.
\end{rem}
\item[b)] The alternative approach proposed in \cite{schenato2014bayesian} and extended here, is to improve the solution using a minimum-variance linear update:
\begin{equation}\label{eq:solSEPFUp}
V_\text{post} = V_\text{prior}+K\left(z_\text{measL} - C_\text{measL}V_\text{prior} \right) \\
\end{equation}
Then the error covariance of estimation $V_\text{post}$ is: 
\begin{equation}\label{eq:sigmapost}
\arraycolsep=1pt
\begin{array}{rl}
\Sigma_\text{post} =&  \Sigma_\text{prior} + K(\Sigma_\text{measL}+C_\text{measL}\Sigma_\text{prior}C_\text{measL}^*)K^* \\[0.1cm]
&-KC_\text{measL}\Sigma_\text{prior}-\Sigma_\text{prior}C_\text{measL}^*K^* 
\end{array}
\end{equation}
and the optimal gain $K$ can be computed minimizing the expected error $\mathbb{E}[(V_\text{post}-V)^*(V_\text{post}-V) ]=\text{tr}(\Sigma_\text{post})$: 
\begin{equation}\label{eq:K}
\arraycolsep=1pt
K = \Sigma_\text{prior}C_\text{measL}^*(C_\text{measL}\Sigma_\text{prior}C_\text{measL}^*+\Sigma_\text{measL})^{-1}
\end{equation}
\end{enumerate}

\begin{prop}
If $V_\text{prior}$ is a feasible solution to \eqref{eq:Icons}, as the solutions in Section \ref{subsec:PF}, then the solution \eqref{eq:solSEPFUp} satisfies \eqref{eq:Icons}. 
\end{prop}
\begin{proof}
For $V_\text{prior}$: $(Y_\text{d})_{\varepsilon,\bullet}V_\text{prior} + (Y_\text{c})_{\varepsilon,\bullet}V_{\text{src}}=0$. Also, since the first term in $K$ is $F$, because $K$ starts with $\Sigma_\text{prior}=F\Sigma_{x_\text{prior}}F^*$, we have $(Y_\text{d})_{\varepsilon,\bullet}K =0$ due to \eqref{eq:F}. Therefore, $V_\text{post}$ from \eqref{eq:solSEPFUp} satisfies \eqref{eq:Icons}.
\end{proof}

\begin{prop}\label{prop:equiv}
If $V_\text{prior}$ is a feasible solution to \eqref{eq:Icons}, then both alternatives \eqref{eq:WLSupdateSol} and \eqref{eq:solSEPFUp} are equal. 
\end{prop}
\begin{proof}
We start from expression \eqref{eq:WLSupdateSol}, using Woodbury's identity \cite{woodbury1950inverting} on the inverse term and \eqref{eq:K} for $K$ we get:
\begin{equation}\label{eq:propequiv1}
\arraycolsep=1pt
\begin{array}{rl}
V_\text{post}=&  F((F^*\Sigma_{\text{prior}}F) - (F^*\Sigma_{\text{prior}}F)F^*C_\text{measL}^*\\[0.1cm] 
&\cdot(C_\text{measL}\Sigma_\text{prior}C_\text{measL}^* + \Sigma_\text{measL})^{-1}C_\text{measL}F(F^*\Sigma_{\text{prior}}F))\\[0.1cm]
& \cdot((F^*\Sigma_{\text{prior}}F)^{-1}x_\text{prior}+F^*C_\text{measL}^*\Sigma_\text{measL}^{-1} \\[0.1cm]
& \cdot(z_\text{measL}-C_\text{measL}V_0))+V_0 \\
=&  (F(F^*\Sigma_{\text{prior}}F)-KC_\text{measL}F(F^*\Sigma_{\text{prior}}F))\\[0.1cm]
& \cdot((F^*\Sigma_{\text{prior}}F)^{-1}x_\text{prior}+F^*C_\text{measL}^*\Sigma_\text{measL}^{-1} \\[0.1cm]
& \cdot(z_\text{measL}-C_\text{measL}V_0))+V_0\\[0.1cm]
=& V_0 + (I_\text{d}-KC_\text{measL})Fx_\text{prior}\\[0.1cm]
& +(I_\text{d}-KC_\text{measL})F(F^*\Sigma_{\text{prior}}F)F^*C_\text{measL}^*\Sigma_\text{measL}^{-1}\\[0.1cm]
 & \cdot (z_\text{measL}-C_\text{measL}V_0))
\end{array}
\end{equation}
Doing some manipulation to the formula of $K$ in \eqref{eq:K} we get:
\begin{equation}
(I_\text{d}-KC_\text{measL})\Sigma_{\text{prior}}C_\text{measL}^*\Sigma_\text{measL}^{-1}=K
\end{equation}
So finally we can convert the expression in \eqref{eq:WLSupdateSol} to \eqref{eq:solSEPFUp}:
\begin{equation*}
\arraycolsep=1pt
\begin{array}{rll}
V_\text{post}=& V_0 +(I_\text{d}-KC_\text{measL})Fx_\text{prior}+K(z_\text{measL}-C_\text{measL}V_0)\\[0.1cm]
= &V_\text{prior}+K(z_\text{measL}-C_\text{measL}V_\text{prior})&\hspace{-0.5cm}\qedhere
\end{array}
\end{equation*}
\end{proof}

The maximum-likelihood estimator of the weighted least-squares solution coincides with the minimum-variance estimator of the linear update, since we are considering Gaussian distributions for the noise. Therefore, it is the same to use one method or the other, when considering only synchronized measurements.

\subsubsection{Updating with unsynchronized measurements}\label{subsec:updatePFNL}~\\
If there are nonlinear measurements, since the magnitude of a complex number as in \eqref{eq:NLmeasMap} is not a holomorphic function, the update step \ref{eq:solSEPFUp} needs to be solved using real variables:
\begin{equation}\label{eq:solSEPFUpNL}
V_\text{rect,post}=  
V_\text{rect,prior}+\tilde{K}
\left[ \begin{array}{c}
\Re\{z_\text{measL} - C_\text{measL}V(V_\text{rect,prior})\} \\
\Im\{z_\text{measL} - C_\text{measL}V(V_\text{rect,prior})\} \\
z_\text{measNL} - C_\text{measNL}(V(V_\text{rect,prior}))
\end{array}  \right] 
\end{equation}
Given that
\begin{equation*}\arraycolsep=1pt
\left[ \begin{array}{c}\Re\{C_\text{measL}V(V_\text{rect,prior})\} \\ \Im\{C_\text{measL}V(V_\text{rect,prior})\} \end{array}  \right] 
\hspace{-0.1cm}=\hspace{-0.1cm}
 \left[\begin{array}{cr}
\Re\{C_\text{measL}\} & -\Im\{C_\text{measL}\} \\
\Im\{C_\text{measL}\} & \Re\{C_\text{measL}\}
\end{array}\right] \hspace{-0.1cm} V_\text{rect,prior}
\end{equation*}
and defining
\begin{equation*}\arraycolsep=1pt
\begin{array}{c}
C_\text{measLNL} = \left[\begin{array}{c}
\left[\begin{array}{cc}
\Re\{C_\text{measL}\} & -\Im\{C_\text{measL}\} \\
\Im\{C_\text{measL}\} & \Re\{C_\text{measL}\}
\end{array}\right] \\
\nabla_{V_\text{rect,prior}} C_\text{measNL}(V(V_\text{rect})) \mid_{V_\text{rect,prior}}
\end{array}\right]  \\[0.5cm]
\Sigma_\text{measLNL} = \text{diag}(\Sigma_\text{rect,measL},\Sigma_\text{measNL})
\end{array}
\end{equation*}
the first-order approximation of the estimation error covariance of this nonlinear filter is: 
\begin{equation}\label{eq:sigmaUpNL}\begin{array}{c}
\Sigma_\text{rect,post} \approx \tilde{K}(\Sigma_\text{measLNL} +C_\text{measLNL}\Sigma_\text{rect,prior}C_\text{measLNL}^T)\tilde{K}^T \\[0.1cm]
+\Sigma_\text{rect,prior}  -\tilde{K}C_\text{measLNL}\Sigma_\text{rect,prior}-\Sigma_\text{rect,prior}C_\text{measLNL}^T\tilde{K}^T
\end{array}
\end{equation}
where now the optimal gain $\tilde{K}$ is: 
\begin{equation*}
\tilde{K} = \Sigma_\text{rect,prior}C_\text{measLNL}^T(C_\text{measLNL}\Sigma_\text{rect,prior}C_\text{measLNL}^T
+\Sigma_\text{measLNL})^{-1}
\end{equation*}
\begin{rem}	
As in Proposition \ref{prop:equiv}, it can be proven that the first-order approximation of the estimation error covariance \eqref{eq:sigmaUpNL} of \eqref{eq:solSEPFUpNL} is equal to the one approximating the estimation error covariance of \eqref{eq:WLSupdate} including nonlinear measurements, computed similarly as in \eqref{eq:sigmaWLS}.
\end{rem}

\section{Case Study}\label{sec:sim}
A simulation of 24 hours with 15 min intervals is run on a test case to compare the different methods for SE. Here we describe the settings of the test case and analyze its results. These algorithms are coded in Python and run on an Intel Core i7-5600U CPU at 2.60GHz with 16GB of RAM.
\subsection{Settings}\label{subsec:simDef}
\begin{itemize}[leftmargin=*]
\item System: We use the 123-bus test feeder \cite{kersting1991radial} available online \cite{testfeeder}, see Fig. \ref{fig:123bus}. This is a challenging example, since it is 3-phase coupled, unbalanced, and larger than other examples in the literature \cite{schenato2014bayesian,kelley1995branch}.
\item Measurements (see Fig. \ref{fig:123bus}): Voltage measurements (red circle for phasor, red square for magnitude only) are placed at buses $79,95,83$ and $300$, current measurements (blue dashed circle for phasor, blue dashed square for magnitude only) at buses $65$ and $48$, and branch current phasor measurements (blue dashed arrow) at branch $150$ (after the regulator) $\to 149$. The standard deviation used to simulate noise in the measurements is $\sigma_\text{meas}=0.01$ according to the standards in \cite{martin2008exploring}. This sensors locations are chosen since they correspond to nodes with big loads and/or towards the end of a feeder, or at a branch transporting a large current. {\color{black} A more extensive method to obtain optimal locations is not within the scope of this paper and will be addressed in future work.}
\item Load Profiles: They are built aggregating over 50 households profiles with maximum loads above $1\,$kW from the dataset available in the DiSC simulation framework \cite{pedersen2015disc}. The values of $S_\text{psd}$ are an average over all households multiplied by 50. This way the loads are similar to the base loads provided in the 123-bus test feeder, and the relative load standard deviation $\sigma_\text{psd}$ is approximately $50\%$. The load profiles can be seen in Fig. \ref{fig:figloads}. 
\item Nomenclature: $V_\text{prior}$ denotes the prior solution \eqref{eq:linapPF}; $V_{\text{WLS}_\text{Ha}}$ and $V_{\text{WLS}_\text{Or}}$ are the traditional Hachtel and Orthogonal solutions of \eqref{eq:WLSpolar} based on weighted least squares \cite{abur2004power}; $V_{\text{WLS}_\text{sub}}$ is the weighted least-squares solution using the subspace \eqref{eq:WLSrect}; and $V_\text{post}$ is the posterior estimate \eqref{eq:solSEPFUpNL}.
\item Accuracy metric: The root mean square error (RMSE) and the maximum absolute error (MaxAE) are determined for the voltages estimated by every method:
{\color{black}
\begin{equation*}\begin{array}{l}
\text{RMSE}_{t,\text{method}_m} = \sqrt{\frac{1}{N} \sum_{i=1}^N \abs{V_{\text{method}_m,t,i}-V_{t,i}}^2}  \\[0.1cm]
\text{maxAE}_{t,\text{method}_m} = \max_{i}\abs{V_{\text{method}_m,t,i}-V_{t,i}}
\end{array}
\end{equation*}
where $V_{t,i},V_{\text{method}_m,t,i}$ denote the value of the $i$-element, at time step $t$ of the simulation, of the actual voltage and the voltage estimated by method $m$ respectively.}
\end{itemize}

\begin{figure}
\centering
\includegraphics[width=9cm,height=7.3cm]{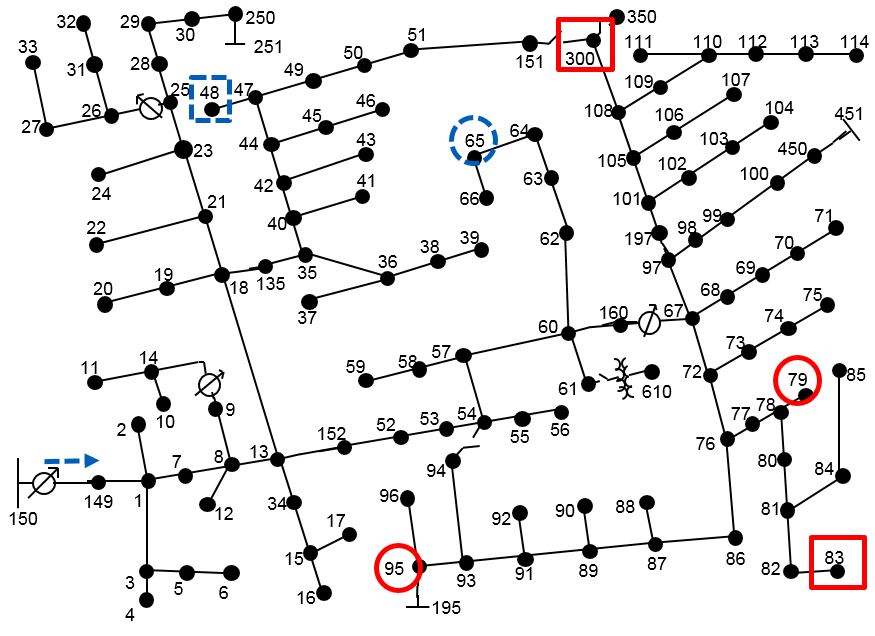}    
\caption{123-bus test feeder with measurements location. The network image has been taken from \cite{testfeeder}.} 
\label{fig:123bus}
\end{figure}

\begin{figure}
\centering
\includegraphics[width=2.66in,height=2in]{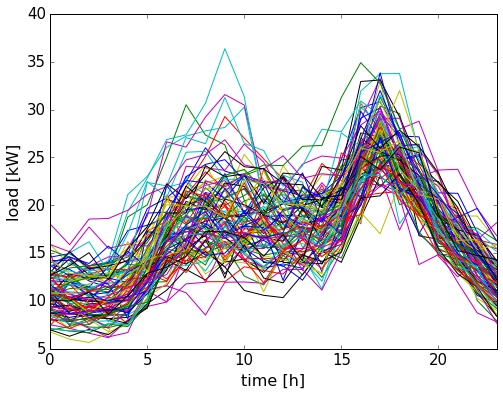}    
\caption{Load profiles are built aggregating over 50 households profiles from the dataset in \cite{pedersen2015disc}.} 
\label{fig:figloads}
\end{figure}

\subsection{Results}\label{subsec:simRes}
Fig. \ref{fig:errorsmethoda} shows how a solution with real-time measurements like $V_\text{post}$ clearly outperforms the prior solution $V_\text{prior}$ using only load estimations. This is relevant in this case, because the $V_\text{prior}$ may produce an error of almost $0.05\,$pu and thus is unable to monitor network constraint violations, while the estimation $V_\text{post}$ produces errors of at most $0.01\,$pu. This shows that even when few measurements are available (in this case only 7 measurements in a 123-bus network), using them can dramatically increase the accuracy of the prediction. Moreover, Fig. \ref{fig:errorsmethodb} shows that the two-step solution $V_\text{post}$ performs statistically as well as all weighted least-squares solutions $V_{\text{WLS}_\text{Ha}}$, $V_{\text{WLS}_\text{Or}}$, $V_{\text{WLS}_\text{sub}}$, {\color{black} with a minimal decrease in accuracy, especially when compare to that of $V_\text{prior}$.} However, the methods using weighted least-squares are computationally much more expensive than the two step method $V_\text{post}$, see Fig. \ref{fig:exectime}. 

To summarize, the solutions including real-time measurements clearly improve the accuracy of the estimation. However, the two-step approach has a much smaller computational cost and is more scalable to bigger networks for real-time operation.


\begin{figure}[!t]
\centering
\subfloat[]{\includegraphics[width=1.45in,height=2in]{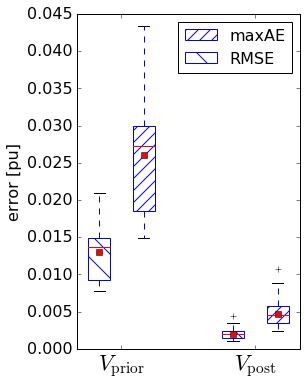}%
\label{fig:errorsmethoda}}
\hfil
\subfloat[]{\includegraphics[width=2in,height=2in]{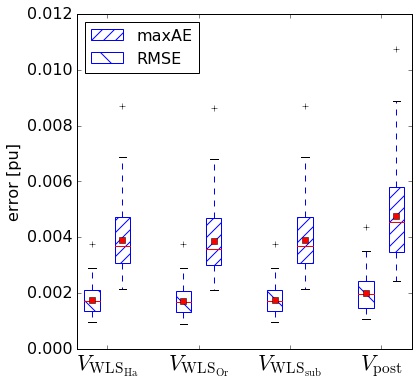}%
\label{fig:errorsmethodb}}
\caption{Box-plots of RMSE and MaxAE for the solutions. The red line indicates the median, the red square the mean.}
\includegraphics[width=2.66in,height=2in]{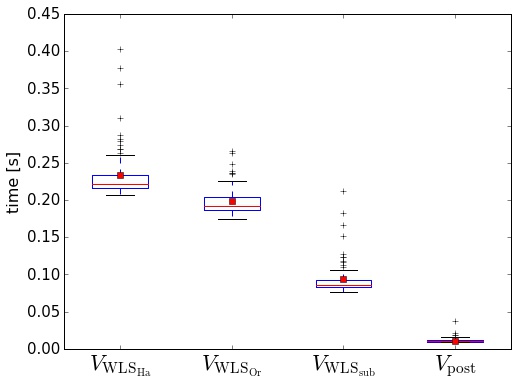}    
\caption{Box-plots of execution times for the solutions. The red line indicates the median, the red square the mean.} 
\label{fig:exectime}
\end{figure}

\section{Conclusions}\label{sec:conc}
We have proposed a methodology for state estimation in large-scale, 3-phase coupled, unbalanced distribution systems, with mixed measurements: phasor-synchronized (linear) and magnitude-unsynchronized (nonlinear), zero load injection buses in the form of constraints, and different sources of noise: magnitude and angle. We have shown that the proposed method is as accurate as standard methods, but computationally more efficient in real-time operation. The computational cost reduction is achieved by two factors: first using a dimensionality reduction to the subspace of feasible solution, and second splitting the problem into a two-step offline and online problem. The iterative offline method estimates the prior solutions, taking most of the computational cost, while the online problem updates this solution in a single step, without requiring further iterations. 

Future work could include adding different sources of distributed generation and investigating how they affect the solutions; developing methodologies to compute the minimum number and optimal allocation of measurement units to guarantee a given accuracy level; and adding dynamic state equations to exploit the historical sensor information.

\section*{Appendix}\label{sec:appendix}
A PMU measurement of a complex number $u=\abs{u}e^{j\theta_u}$ with magnitude and angle noise can be expressed as:
\begin{equation*}\arraycolsep=1pt\begin{array}{rl}
\tilde{u} = & (\abs{u}+\omega_\text{mag})e^{j(\theta_u+\omega_\theta)} \\[0.1cm]
= & (\abs{u}+\omega_\text{mag})(\cos(\theta_u+\omega_\theta)+j\sin(\theta_u+\omega_\theta)) 
\end{array}
\end{equation*}
where $\omega_\text{mag}\sim \mathcal{N}(0,1\%\abs{u})$ (or $\tilde{\omega}_\text{mag}=\frac{\omega_\text{mag}}{\abs{u}}\sim \mathcal{N}(0,0.01)$), $\omega_\theta \sim \mathcal{N}(0,0.01\text{ rad})$, according to the standards in \cite{martin2008exploring}. Using trigonometric identities, considering that $\abs{\omega_\theta} \ll 1$ so that $\sin(\omega_\theta)\approx \omega_\theta$, $\cos(\omega_\theta)\approx 1$, and neglecting second-order noise terms, we obtain:
\begin{equation*}\arraycolsep=1pt
\begin{array}{rl}
\tilde{u} \approx & \abs{u}(\cos(\theta_u)+j\sin(\theta_u)) + \omega_\text{mag}(\cos(\theta_u)+j\sin(\theta_u))\\ & + \omega_\theta\abs{u}(-\sin(\theta_u)+j\cos(\theta_u)) \\
\approx & u + \tilde{\omega}_\text{mag}u +j \omega_\theta u = u + u\left(\tilde{\omega}_\text{mag}+j\omega_\theta\right)
\end{array}
\end{equation*}

\bibliographystyle{IEEEtran}
\bibliography{ifacconf}

\begin{thebibliography}{10}
\providecommand{\url}[1]{#1}
\csname url@samestyle\endcsname
\providecommand{\newblock}{\relax}
\providecommand{\bibinfo}[2]{#2}
\providecommand{\BIBentrySTDinterwordspacing}{\spaceskip=0pt\relax}
\providecommand{\BIBentryALTinterwordstretchfactor}{4}
\providecommand{\BIBentryALTinterwordspacing}{\spaceskip=\fontdimen2\font plus
\BIBentryALTinterwordstretchfactor\fontdimen3\font minus
  \fontdimen4\font\relax}
\providecommand{\BIBforeignlanguage}[2]{{%
\expandafter\ifx\csname l@#1\endcsname\relax
\typeout{** WARNING: IEEEtran.bst: No hyphenation pattern has been}%
\typeout{** loaded for the language `#1'. Using the pattern for}%
\typeout{** the default language instead.}%
\else
\language=\csname l@#1\endcsname
\fi
#2}}
\providecommand{\BIBdecl}{\relax}
\BIBdecl

\bibitem{abur2004power}
A.~Abur and A.~G. Exposito, \emph{Power System State Estimation: Theory and
  Implementation}.\hskip 1em plus 0.5em minus 0.4em\relax CRC Press, 2004.

\bibitem{huang2012state}
Y.-F. Huang, S.~Werner, J.~Huang, N.~Kashyap, and V.~Gupta, ``State estimation
  in electric power grids: Meeting new challenges presented by the requirements
  of the future grid,'' \emph{IEEE Signal Processing Magazine}, vol.~29, no.~5,
  pp. 33--43, 2012.

\bibitem{monticelli2000electric}
A.~Monticelli, ``Electric power system state estimation,'' \emph{Proceedings of
  the IEEE}, vol.~88, no.~2, pp. 262--282, 2000.

\bibitem{hayes2014state}
B.~Hayes and M.~Prodanovic, ``State estimation techniques for electric power
  distribution systems,'' in \emph{2014 IEEE European Modelling Symposium
  (EMS)}, 2014, pp. 303--308.

\bibitem{holten1988comparison}
L.~Holten, A.~Gjelsvik, S.~Aam, F.~F. Wu, and W.-H. Liu, ``Comparison of
  different methods for state estimation,'' \emph{IEEE Transactions on Power
  Systems}, vol.~3, no.~4, pp. 1798--1806, 1988.

\bibitem{stott1974fast}
B.~Stott and O.~Alsa{\c{c}}, ``Fast decoupled load flow,'' \emph{IEEE
  Transactions on Power Apparatus and Systems}, no.~3, pp. 859--869, 1974.

\bibitem{garcia1979fast}
A.~Garcia, A.~Monticelli, and P.~Abreu, ``Fast decoupled state estimation and
  bad data processing,'' \emph{IEEE Transactions on Power Apparatus and
  Systems}, no.~5, pp. 1645--1652, 1979.

\bibitem{madani2011pmu}
V.~Madani, M.~Parashar, J.~Giri, S.~Durbha, F.~Rahmatian, D.~Day, M.~Adamiak,
  and G.~Sheble, ``{PMU} placement considerations - {A} roadmap for optimal
  {PMU} placement,'' in \emph{Power Systems Conference and Exposition
  (PSCE)}.\hskip 1em plus 0.5em minus 0.4em\relax IEEE, 2011, pp. 1--7.

\bibitem{baran1994state}
M.~E. Baran and A.~W. Kelley, ``State estimation for real-time monitoring of
  distribution systems,'' \emph{IEEE Transactions on Power Systems}, vol.~9,
  no.~3, pp. 1601--1609, 1994.

\bibitem{kelley1995branch}
M.~B.~A. Kelley, ``A branch current based state estimation method for
  distribution systems,'' \emph{IEEE Transactions on Power Systems}, vol.~10,
  pp. 483--491, 1995.

\bibitem{lin2001highly}
W.-M. Lin, J.-H. Teng, and S.-J. Chen, ``A highly efficient algorithm in
  treating current measurements for the branch-current-based distribution state
  estimation,'' \emph{IEEE Transactions on Power Delivery}, vol.~16, no.~3, pp.
  433--439, 2001.

\bibitem{schenato2014bayesian}
L.~Schenato, G.~Barchi, D.~Macii, R.~Arghandeh, K.~Poolla, and A.~V. Meier,
  ``Bayesian linear state estimation using smart meters and pmus measurements
  in distribution grids,'' in \emph{2014 IEEE International Conference on Smart
  Grid Communications (SmartGridComm)}, Nov 2014, pp. 572--577.

\bibitem{zhou2006alternative}
M.~Zhou, V.~A. Centeno, J.~S. Thorp, and A.~G. Phadke, ``An alternative for
  including phasor measurements in state estimators,'' \emph{IEEE Transactions
  on Power Systems}, vol.~21, no.~4, pp. 1930--1937, 2006.

\bibitem{hu2011belief}
Y.~Hu, A.~Kuh, T.~Yang, and A.~Kavcic, ``A belief propagation based power
  distribution system state estimator,'' \emph{IEEE Computational Intelligence
  Magazine}, vol.~6, no.~3, pp. 36--46, 2011.

\bibitem{bolognani2014state}
S.~Bolognani, R.~Carli, and M.~Todescato, ``State estimation in power
  distribution networks with poorly synchronized measurements,'' in \emph{2014
  IEEE 53rd Annual Conference on Decision and Control}.\hskip 1em plus 0.5em
  minus 0.4em\relax IEEE, 2014, pp. 2579--2584.

\bibitem{guo2013efficient}
Y.~Guo, W.~Wu, B.~Zhang, and H.~Sun, ``An efficient state estimation algorithm
  considering zero injection constraints,'' \emph{IEEE Transactions on Power
  Systems}, vol.~28, no.~3, pp. 2651--2659, 2013.

\bibitem{martin2008exploring}
K.~Martin, D.~Hamai, M.~Adamiak, S.~Anderson, M.~Begovic, G.~Benmouyal,
  G.~Brunello, J.~Burger, J.~Cai, B.~Dickerson \emph{et~al.}, ``{Exploring the
  IEEE standard C37. 118--2005 synchrophasors for power systems},'' \emph{IEEE
  Transactions on Power Delivery}, vol.~23, no.~4, pp. 1805--1811, 2008.

\bibitem{korres2010robust}
G.~N. Korres, ``A robust algorithm for power system state estimation with
  equality constraints,'' \emph{IEEE Transactions on Power Systems}, vol.~25,
  no.~3, pp. 1531--1541, 2010.

\bibitem{simoes1981robust}
A.~Simoes-Costa and V.~Quintana, ``A robust numerical technique for power
  system state estimation,'' \emph{IEEE Transactions on Power Apparatus and
  Systems}, no.~2, pp. 691--698, 1981.

\bibitem{bolognani2016existence}
S.~Bolognani and S.~Zampieri, ``On the existence and linear approximation of
  the power flow solution in power distribution networks,'' \emph{IEEE
  Transactions on Power Systems}, vol.~31, no.~1, pp. 163--172, 2016.

\bibitem{woodbury1950inverting}
M.~A. Woodbury, ``Inverting modified matrices,'' \emph{Statistical Research
  Group Memorandum Reports. Princeton, NJ: Princeton University}, no.~42, 1950.

\bibitem{kersting1991radial}
W.~Kersting, ``Radial distribution test feeders,'' \emph{IEEE Transactions on
  Power Systems}, vol.~6, no.~3, pp. 975--985, 1991.

\bibitem{testfeeder}
D.~T. Feeders, ``{IEEE PES Distribution Systems Analysis Subcommittee's Radial
  Test Feeders},'' \url{ewh.ieee.org/soc/pes/dsacom/testfeeders.html}, 1991,
  [Online]. Accessed: 2016-10-20.

\bibitem{pedersen2015disc}
R.~Pedersen, C.~Sloth, G.~B. Andresen, and R.~Wisniewski, ``{DiSC}: A
  simulation framework for distribution system voltage control,'' in
  \emph{European Control Conference (ECC)}, 2015, pp. 1056--1063.

\end{thebibliography}



%

\end{document}